\documentclass[11pt]{article}
\usepackage{comment}
\usepackage{etex}

\usepackage{graphicx, booktabs, pstricks, floatrow, caption, multirow, wrapfig, amsthm, amssymb, amsmath}
\usepackage[utf8]{inputenc}

\input prepictex
\input postpictex
\input pictexwd

\usepackage{graphicx}
\usepackage{amsthm,amsmath}
\usepackage{amssymb}
\usepackage[ruled,vlined]{algorithm2e}
\usepackage{float}
\usepackage{paralist}
\usepackage{xcolor}
\usepackage{url}

\usepackage[affil-it]{authblk}
\usepackage[round]{natbib}

\usepackage{todonotes}

\newtheorem{theorem}{Theorem}

\newtheorem{lemma}[theorem]{Lemma}
\newtheorem{proposition}[theorem]{Proposition}
\newtheorem{example}{Example}
\newtheorem{definition}{Definition}

\newtheorem{assumption}{Assumption}

\newcommand{\calR}{\mathcal{R}}
\newcommand{\tp}{\mathrm{top}}
\newcommand{\scr}{\mathrm{sc}}
\newcommand{\GS}{\mathcal{GS}}
\newcommand{\GSP}{\mathcal{GSP}}

\title{\bf Gibbard--Satterthwaite Games\\ for $k$-Approval Voting Rules}

\author{Umberto Grandi}
\affil{University of Toulouse, umberto.grandi@irit.fr}

\author{Daniel Hughes}
\affil{The University of Auckland, dhug729@aucklanduni.ac.nz}

\author{Francesca Rossi}
\affil{University of Padova, frossi@math.unipd.it}

\author{Arkadii Slinko}
\affil{The University of Auckland, a.slinko@auckland.ac.nz}

\date{}

\begin{document}

\maketitle

\begin{abstract}
The Gibbard-Satterthwaite theorem implies the existence of voters, called manipulators, who can change the election outcome in their favour by voting strategically. When a given
preference profile admits several such manipulators, voting becomes a game played by these voters, who have to reason strategically about each others' actions. To complicate the game even further, counter-manipulators may then try to counteract the actions of manipulators. Our voters are boundedly rational and do not think beyond manipulating or countermanipulating. We call these games Gibbard--Satterthwaite Games. In this paper we look for conditions that guarantee the existence of a Nash equilibria in pure strategies. 
\end{abstract}

\section{Introduction}
Voting is a common method of preference aggregation, which enables the participating agents 
to identify the best candidate given the individual agents' rankings of the candidates. 
However, no ``reasonable'' voting rule is immune to manipulation: as shown by \citet{gib:j:gs}
and \citet{sat:j:gs}, if there are at least three candidates, then any onto, non-dictatorial 
voting rule admits a preference profile (a collection of voters' rankings) where some voter would be better
off by submitting a ranking that differs from his truthful one {\color{black} or, in other words, his truthful vote is not the best response to the votes of other voters}. We call such voter a Gibbard-Satterthwaite manipulator or GS-manipulator for short. When such a manipulator is unique, he\footnote{In the paper we refer to candidates as females and voters as males.} then has a disproportional influence on the election outcome.
However, in the presence of multiple manipulators their attempt to manipulate the election simultaneously in an uncoordinated fashion (and we assume that no coordination devices exist) may bring an outcome that differs not just from the outcome under the truthful voting, but also from the outcome
that any of the GS-manipulators could anticipate. This may be due to possible complex interference
among the different manipulative votes, and may deter some of them (especially risk-averse ones) from manipulating. When we also include in consideration those voters who cannot manipulate themselves but can prevent others from manipulating---the so-called countermanipulators---situation becomes even more complex and can be described only in game-theoretic terms.  Let us illustrate this by an example.

\begin{example}
\label{ex1}
In Table~\ref{table:intro} we describe a voting situation. There are four voters and five candidates, and each voter ranks the candidates from the most to the least preferred. Suppose that we use Plurality rule  to compute the winner, which declares the candidate with the highest number of first positions the winner. Ties are resolved following a predetermined ordering over the candidates, in this case $w>a>b>c>d$.

\begin{table}[h]
\begin{center}
\begin{tabular}{ccccc}
voter 1&voter 2&voter 3&voter 4\\
\midrule
$a$&$b$&$w$&$d$\\
$b$&$a$&$c$&$w$\\
$c$&$c$&$a$&$a$\\
$d$&$d$&$d$&$b$\\
$w$&$w$&$b$&$c$\\
\end{tabular}
\caption{A preference profile. The most preferred candidates are on top, followed by the less preferred candidates in a complete ranking.
}\label{table:intro}
\end{center}
\end{table}

In the situation described in Table~\ref{table:intro} the winner is $w$, thanks to the tie-breaking. This result is the worst possible for the first two players, who each have the possibility of manipulating the result in  favour of $b$ and $a$, respectively.  They are the only GS-manipulators at this profile, and let us call their sincere strategies $s_1,s_2$ and insincere $i_1,i_2$, respectively.   Let us zoom in on the situation when: voter 1 and voter 2 are strategic and other voters are not.  Voter 1's insincere strategy consists in voting for $b$ instead of $a$ and make $b$ the winner, voter 2's can vote insincerely in favour of $a$ instead of $b$ and make $a$ the winner. If both manipulate at the same time, their efforts will cancel out. They are playing an anti-coordination game.  We can represent this game with numbers being the positions of a winning candidate in the individual ranking, 0 for the least preferred and 4 for the top one:

\[
\begin{array}{c||c|c|}
 &s_2&i_2\\
 \hline
 \hline
s_1&0,0&4,3\\
\hline
i_1&3,4&0,0\\
\hline
\end{array}
\] 

Let us now put spotlight on all voters.
We observe that voter 3 is happy and does not have reasons for strategising. Voter 4 does not have any incentive to manipulate: the current winner $w$ is in his second position, so giving her more support will not change the outcome. Hence voting for $w$ instead of $d$ is not a manipulation for him. However, this move is a very strong countermanipulation; if voter 4 fears any strategic move from any of the first three players: giving additional support to $w$  makes manipulation impossible, ending any strategic considerations.
\end{example}

We see that, even for such a simple voting rule as Plurality, a single profile can give us a plethora of games depending on which voters are strategic and which are not. A non-strategic voter has only his sincere vote in his strategy set, while a strategic voter has more than one strategy. We are interested in the properties of the normal-form games that arise under $k$-Approval voting rules (and Plurality is 1-Approval). These rules are simple enough to allow for a classification of voting manipulations, but complex enough to admit the realization of non-trivial games. 

As is usually the case, in the initial investigation, like this one, it is customary to assume the full information framework which means that everybody's sincere preferences are publicly known as well as their strategy sets.\footnote{A similar approach is taken, for example, in the investigation of the games that appear in generalised second price auctions by \cite{EOS2007}.} However, the voting intentions of the voters remain private to those voters. 

An important novel feature of the games, considered in this paper, which distinguishes them from voting games that have been considered in prior literature (see Section~\ref{sec:related} for related literature survey), is exactly the introduction of types of players which are characterised by their strategy sets. There are several reasons for the introduction of types. Firstly, reducing players' strategy sets  we can secure their bounded rationality. The second is that the knowledge of the sincere profile does not allow to unambiguously decide who is strategic and who is not.  Voters may be able to manipulate but reject this on moral grounds or they may be unable to calculate their manipulation.  On the other hand, a voter may not be able to manipulate but can take preventive measure from a disastrous (for him) effect of someone else's manipulation (like voter~4 in Example~\ref{ex1}).  Thus the introduction of strategy sets  allows us to bring into a spotlight and to study in isolation various aspects of strategic manipulation, e.g., the interaction of Gibbard-Satterthwaite manipulators (e.g., voters 1,2 and 3 in Example~\ref{ex1}) or the interaction of a manipulator and a countermanipulator etc. 

We can summarize this as follows: given a voting rule, every profile of voters' preferences gives rise to a number of games that can be played;  we will call them {\em Gibbard-Satterthwaite games or GS-games}. These games will differ by the set of strategies available to the players and, in particular, by the division of voters into strategic and non-strategic ones.  We treat the strategy sets as a state of nature whose move makes the structure of the game a public knowledge.

The simplest non-trivial example of our framework involves two players each having two strategies: one sincere and one insincere, we call them $2$-by-$2$ games. They can be both GS-manipulators or, alternatively, one can be a GS-manipulator and another a countermanipulator. It does not mean that the election from which this game arises has two voters only, simply only two voters at the given profile are strategic as determined by nature.

We ask whether any $2$-by-$2$ game can be represented as a GS-game. To answer this question
we need a classification of $2$-by-$2$ games. As the existing classifications turn out to be 
too fine-grained for our purposes, we develop a simple coarser classification, and observe
that the definition of GS-games imposes certain restrictions on players' preferences. Combining
this observation with symmetry arguments, we arrive at $6$ basic types of $2$-by-$2$ games played by two manipulators.
We then show that, while all six games can be obtained as GS-games under the $2$-Approval voting
rule, but for Plurality rule (1-Approval) only four of them are realisable. We also obtain a similar classification of $2$-by-$2$ manipulator and countermanipulator games. 

For GS-games with more than two players we bring under the spotlight the situation when all players are GS-manipulators. 
We study the existence of pure strategy Nash equilibria in such games. We show that every GS-game for Plurality has a Nash equilibrium, and identify necessary and sufficient conditions for the existence of Nash equilibria for $2$-Approval games. It appears that a mild rationality condition which we call Soundness Assumption is sufficient and we construct a $2$-Approval game with no Nash equilibria. We also found sufficient conditions for the existence of Nash equilibria of $3$-Approval games. These sufficient conditions assume that manipulating voters choose manipulating strategies which are in some sense minimal. However, we show that  this Minimality Assumption fail to ensure the existence of Nash equilibria for $4$-Approval games.

The paper is organised as follows. In Section~\ref{sec:related} we discuss related work, and in Sections~\ref{sec:prelim} and \ref{sec:model} we introduce Gibbard-Sattethwaite manipulation games. The main contributions of the paper are presented in Section~\ref{sec:2by2}, in which we classify 2-by-2 manipulation games, and in Section~\ref{sec:nash}, where we study the existence of Nash equilibria in arbitrarily large games for $k$-approval. Sections~\ref{sec:discussion} and~\ref{sec:conclusions} discuss the results presented and conclude the paper.
 
\section{Related Work}\label{sec:related}


There is a substantial body of research dating back to \citet{far:b:voting} 
that explores the consequences of modeling non-truthful voting as a strategic game; 
see, e.g., \cite{mou:j:dominance,fed-sen-wri:j:entry,mye-web:j:voting,deSinopoli,dhi-loc:j:dominance,sertel2004strong,de2015stable}. The most popular framework so far has been the one introduced by \cite{mye-web:j:voting}. This model, in particular, stipulate that each voter has a utility for the election of each candidate (so it is not a purely ordinal in nature). Myerson and Weber suggested the use of Nash equilibria and other solution concepts for the analysis of voting games, however, sometimes this idea led to a large number weird Nash equilibria, some Pareto dominated. In further works many attempts have been made to weed them out. The following methods were considered: equilibria refinements \citep{deSinopoli}, costly voting \citep{sin-ian:j:costly}, generic utilities \citep{de2015stable}. The problem however remains not completely solved.  To the best of our knowledge, in all of these papers the set of players consists of all voters, i.e., a player is allowed to vote non-truthfully even if  he would be unable to manipulate the election on his own or countermanipulate. This leads to a large number weird Nash equilibria, some Pareto dominated. Restricting the set of players to GS-manipulators in the original profile, as we do in this paper,  alters the problem substantially; for instance, it rules out ``bad'' Nash equilibria such as when all players vote for the same undesirable candidate. The problems with this model primarily stem from the fact that voters' are allowed to vote irrationally.

In contrast,  \cite{EGRS} assume that the voters reason about potential actions of other voters assuming that they are boundedly rational and they use an adaptation of the cognitive hierarchy model for this.
They take non-strategic (sincere) voters as those belonging to level 0. The players of level 1 give best responce assuming that all other players belong to level 0, and, in case if this responce is not unique (i.e., when this particular voter is not a Gibbard-Satthethwaite manipulator) it is defined to be the sincere vote. The players of level 2 give their best responce to assuming that all other players belong to level 0 or level 1. We note that players of level 2 are already quite sophisticated. They can, for example, think of countermanipulating or they can strategically stay sincere when they can manipulate. The emphasis there is on the complexity of a level 2 voter deciding  whether his manipulative strategy weakly dominates his sincere strategy. They present a polynomial time algorithm for 2-Approval but prove NP-hardness for 4-Approval voting rule.

The algorithmic aspects of voting games have recently received some attention as well \cite{des-elk:c:eq,xia-con:c:spne,tho-lev-ley:c:empirical,obr-mar-tho:c:truth-biased}.  

Iterative voting is the closest topic to this paper considered in the literature. In this model players change their votes one by one in response to the current outcome
\cite{mei-pol:c:convergence,lev-ros:c:iterative,rei-end:c:polls,rey-wil:c:bestreply}. However, there are significant differences with our framework. The main one is that a voting manipulation game is a one-shot game while in the iterative voting a player can make several moves. We have a fixed strategy set for each player while in the iterative voting players decide on their next move depending on the profile that resulted after the previous moves. So the strategy sets of players change over time. The common feature is that both approaches assume boundedly rational voters.\footnote{From the point of view of the cognitive hierarchy model all their voters belong to level-1 of the hierarchy.}

\cite{bar-coe:j:non-controversial-k-names} introduce a totally different type of voting games  where the players choose by voting a subset of candidates with a fixed size from a given set of candidates and an external actor, Chooser, selects one of the preselected candidates. They call this game the Random Chooser Game and also use Nash equilibria to analyse these games.



\section{Preliminaries}\label{sec:prelim}
We consider elections over a candidate set $C=\{c_1, \dots, c_m\}$ in which $n$ voters $1,2,\ldots,n$ participate.
An election is defined by a {\em preference profile} $V=(v_1, \dots, v_n)$, where each $v_i$, 
$i=1, \dots, n$, is a total order over $C$; we refer to $v_i$ as the {\em vote}, or {\em preferences}, of voter $i$.
For two candidates $c_1, c_2\in C$ we write $c_1\succ_i c_2$ if voter~$i$ ranks $c_1$ above $c_2$;
if this is the case, we say that voter $i$ {\em prefers} $c_1$ to $c_2$.
For brevity we will sometimes write $ab\dots z$ to represent a vote $v_i$
with $a\succ_i b\succ_i\cdots\succ_i z$.
We denote by $\tp(v_i)$ the top candidate in $v_i$. 
Also, we denote by $\tp_k(v_i)$
the set of top $k$ candidates in~$v_i$.

Given a preference profile $V=(v_1, \dots, v_n)$, we denote by $(V_{-i}, v'_i)$
the preference profile obtained from $V$ by replacing $v_i$ with $v'_i$; for readability,
we will sometimes omit the parentheses around $(V_{-i}, v'_i)$ and write $V_{-i}, v'_i$.

Let $X = (x_1, \dots, x_\ell)$ and $Y = (y_1, \dots, y_\ell)$ be two sequences over disjoint sets of candidates such that no candidate is repeated in any of them and $v$ is a vote.
Then $v[X; Y]$ denotes the vote obtained
by swapping $x_j$ with $y_j$ for $j=1, \dots, \ell$ in the individual preference ordering $v$. We often denote sequences as $X=x_1\ldots x_\ell$ and $Y=y_1\ldots y_\ell$. Then we write 
$v[x_1\ldots x_\ell; y_1\ldots y_\ell]$ instead of $v[X; Y]$.



A (resolute) {\em voting rule} is a mapping $\calR$ that, given a profile $V$, outputs a candidate $\calR(V)\in C$ called the {\em winner} at $V$ under $\calR$. We say that two votes $v$ and $v'$ of voter $i$ over the same candidate set $C$
are {\em equivalent} with respect to a voting rule $\calR$,
if $\calR(V_{-i}, v)=\calR(V_{-i}, v')$ for every profile~$V$.

In this paper we  consider $k$-Approval voting rules. Under $k$-Approval, $1\le k\le m-1$, each candidate receives one point from
each voter who ranks her in top $k$ positions. The candidate(s) with the highest score wins. 
Since any $k$-Approval voting rule is not resolute and two or more candidates can share the highest score, we complement it with a tie-breaking. In this paper any ties that occur are broken according to a fixed order $>$, usually alphabetic, over $C$.\footnote{Our results can be adapted to any other natural tie-breaking.} 

  We denote the $k$-Approval score of a candidate $c$ in a profile $V$ by $\scr_k(c, V)$. 
We will sometimes denote the $k$-Approval rule by $k$-App. $1$-Approval is also known as Plurality.   It is easy to see that $v$ and $v'$ are equivalent with respect to $k$-Approval if and only if $\tp_k(v)=\tp_k(v')$.
Occasionally we refer to the Borda rule. Under this rule, each candidate gets $m-j$ points from each voter who
ranks her in position $j$.


\section{The Model}\label{sec:model} 

{\color{black} Our goal is to model and investigate situations that arise in voting when one or more voters are strategic. 
We model such situations as normal form games that we call {\em voting manipulation games}.
We will now define such situations and games formally. We start with defining the first and the main type of a strategic voter.
}


\begin{definition}
\label{GS-manipulator}
We say that a voter $i$ is a {\em Gibbard--Sat\-ter\-thwaite manipulator}, or a {\em GS-manipula\-tor}, 
at a profile $V=(v_1, \dots, v_n)$ with respect to a voting rule $\calR$, 
if there exists a vote $v'_i\neq v_i$ such that $i$ strictly prefers $\calR(V_{-i}, v'_i)$ to $\calR(V)$.  If $\calR(V_{-i}, v'_i)=p$, we will also say that $i$ manipulates {\em in favor of $p$}.
\end{definition}

\begin{definition}
\label{GS-manipulation}
 A vote $v'_i$ is called a {\em GS-manipulation} of voter $i$ if
\begin{itemize} 
\item 
$i$ prefers $\calR(V_{-i}, v'_i)$ to $\calR(V)$, and, 
\item 
for every $v''_i$ it holds that either $\calR(V_{-i}, v'_i) = \calR(V_{-i}, v''_i)$
or $i$ prefers $\calR(V_{-i}, v'_i)$ to $\calR(V_{-i}, v''_i)$.
\end{itemize}
\end{definition}

Sometimes a voter can manipulate in favor
of several different candidates; however, in Definition~\ref{GS-manipulation} 
we require the voter to focus on his most preferred candidate among the ones
he can make the election winner. This is a mild rationality assumption.
%

Below is another important type of strategic voter.

\begin{definition}
\label{GS-countermanipulator}
Suppose voter $i$ is a GS-manipulator at $V=(v_1, \dots, v_n)$ with respect to a voting rule $\calR$ and his manipulation is $v_i'$. We say that a voter $j$ is a {\em countermanipulator} at a profile $V$ against $v_i'$  if
\begin{itemize} 
\item there exists a vote $v'_j\neq v_j$ such that $j$ prefers $\calR((V_{-i}, v'_i)_{-j},v_{j}')$ to $\calR(V_{-i}, v'_i)$, and, 
\item 
for every $v''_j$ it holds that either $\calR((V_{-i}, v'_i)_{-j},v_{j}')=\calR((V_{-i}, v'_i)_{-j},v_{j}'')$
or $j$ prefers \newline $\calR((V_{-i}, v'_i)_{-j},v_{j}')$ to $\calR((V_{-i}, v'_i)_{-j},v_{j}'')$.
\end{itemize}
In this case vote $v'_j$ is called a {\em countermanipulation} of voter $j$ against $v'_i$.
\end{definition}

{\color{black} Of course these are two most basic types of strategic voters. There are more sophisticated ones. For example, a voter may not be able to change the result of the election unilaterally (in fact large elections voters are seldom pivotal) but he may hope that there will be other likeminded voters who will also change their vote in a similar way and the desired change may come about as a result of combined efforts \citet{safe1,safe2}. Alternatively, he may countermanipulate against a coalition of manipulaters, etc. The higher his level of rationality, the more strategic motives the voter understands and the more complex game he faces as a result \cite{EGRS}. 
}

We denote the set of all strategic voters at a profile $V$ with respect to a voting rule $\calR$
by $N(V, \calR)$, {\color{black} after nature makes its move, this set is known to everybody}. 

Recall that a {\em normal-form game} is defined by a set of {\em players} $N$, and, for each player $i\in N$,
 a set of {\em actions} $A_i$ and a preference relation $\succeq_i$ defined on the space of 
{\em action profiles}, i.e., on tuples of the form $(a_1, \dots, a_n)$, where $a_i\in A_i$ for all $i\in N$. 
 Alternatively we can think that there is a function $f\colon A_1\times\ldots\times A_n \to \mathcal{O}$, where $\mathcal{O}$ is the set of outcomes and relations $\succeq_i$ are defined on $\mathcal{O}$.\footnote{While  normal-form games may be defined either in terms of utility functions or in terms of
preference relations, the latter approach is more suitable for our setting, as we only have ordinal
information about the voters' preferences.}
In our case, the preference relation of player $i$ on the action profiles is determined
by the outcome of $\calR$ on those action profiles.

{\color{black} A voting manipulation game at a profile $V$ is any game with the set of players $N=N(V, \calR)$ such that for every player $i$ his strategy set $A_i$, includes this voter's sincere vote.}
Given an action profile $V^* = (v^*_i)_{i\in N}$, let 
$V[V^*]=(v'_1, \dots, v'_n)$ be the preference profile
such that $v'_i=v_i$ for $i\not \in N$ and $v'_i=v^*_i$ for $i\in N$.
Then, given two action profiles $V^*$ and $V^{**}$, 
we write $V^*\succeq_i V^{**}$ if and only if $\calR(V[V^*])=\calR(V[V^{**}])$ or player $i$
prefers $\calR(V[V^*])$ to $\calR(V[V^{**}])$. {\color{black} In what follows, we denote this game $G=(V, \calR, (A_i)_{i\in N(V, \calR)})$.}

As we found, inclusion of countermanipulators as strategic players, beyond $2$-by-$2$ games, makes the analysis 
of the game too complex. Also, a high degree of rationality must be assumed from the players. As is shown in \cite{EGRS} 
a player must be at least in Level 2 of the cognitive hierarchy to become a countermanipulator. So most of the time we assume that 
only GS-manipulators are strategic. Such games we call {\em Gibbard-Satterthwaite games} or {\em GS-games}.\par

For each preference profile $V$ and each voting rule $\calR$, a GS-game is a normal-form game 
defined as follows. For each game in this family, the set of players $N$ is the set of all GS-manipulators 
in $V$ under $\calR$. For each player $i$, his set of actions $A_i$ consists of his truthful vote
and a subset (possibly empty) of his GS-manipulations; different choices of these subsets correspond 
to different games in the family. 
We 
denote the set of all GS-games for $V$ and $\calR$
by $\GS(V, \calR)$. Note that all games in $\GS(V, \calR)$ have the same set of players, namely, $N(V, \calR)$,
so an individual game in $\GS(V, \calR)$ is fully determined by the players' sets of actions, i.e.,
$(A_i)_{i\in N(V, \calR)}$. 
When $V$ and $\calR$ are clear from the context, we simply write $G = (A_i)_{i\in N}$.
We refer to an action profile in a GS-game as a {\em GS-profile};
we will sometimes identify the GS-profile $V^* = (v^*_i)_{i\in N}$ with the preference profile
$V[V^*]$. We denote the set of all GS-profiles in a game $G$ by $\GSP(G)$. 

We emphasise that, rather than considering games where each player's set of actions
consists of his truthful vote and {\em all} of his GS-manipulations, 
we allow the players to limit themselves to subsets of their GS-manipulations.  
There are several reasons for that. First, the space of all GS-manipulations
for a given voter can be very large, and a player may be unable or unwilling to identify all 
such votes; indeed, even counting the number of GS-manipulations for a given
voter is a non-trivial computational problem 
\cite{bac-bet-fal:c:counting}.
Thus, the player may use a specific algorithm (e.g., greedy algorithm of \cite{btt89} for the class of scoring rules)
to find his GS-manipulation; in this case, his set of actions would consist of his truthful
vote and the output of this algorithm. Also, the player may choose to ignore GS-manipulations
that are (weakly) dominated by other GS-manipulations. 
Finally, a player may prefer not to change his vote beyond 
what is necessary to make his target candidate the election winner, either because 
he wants his vote to be as close to his true preferences as possible (see Obraztsova and Elkind, \citeyear{obr-elk:c:opt}), or for fear of
unintended consequences of such changes in the complex environment of the game.

\section{2-by-2 Voting Manipulation Games}\label{sec:2by2}

In this section, we investigate which $2$-by-$2$ games (i.e., games with two players, and two actions per player)
can be represented as GS-games or Manipulator/Countermanipulator games. Our goal is to show that even in the so restricted framework we can realise 
a surprising variety of games.

\subsection{Representation of 2-player games}

To answer the question which $2$-by-$2$ games can be realised as voting manipulation games of a particular kind, we need a suitable classification
of $2$-by-$2$ games. Note, first, that every such game corresponds to $4$ action profiles, 
and is fully described by giving both players' preferences over these profiles.
By considering all possible pairs of preference relations over domains of size $4$,
Fraser and Kilgour (\citeyear{FK1986}) show that there are $724$ distinct $2$-by-$2$ games. However,
this classification is too fine-grained for our purposes. Thus, we propose a simplified
approach that is based on the following two principles. First, we only compare action 
profiles that differ in exactly one component. Second, when comparing two profiles
that differ in the $i$-th component ($i=1, 2$), we only take into account the preferences 
of the $i$-th player. Thus, every $2$-by-$2$ game can be represented by a diagram
with $4$ vertices and $4$ directed edges, where an edge is directed from a less
preferred profile to a more preferred profile and a bidirectional edge indicates indifference. 

\subsection{Two manipulators GS-game}

Now, let us focus on GS-games with $2$ players and $2$ actions per player, one of which is their sincere vote. 
For each player, let $s$ denote his sincere vote and let $i$ denote his manipulative vote;
thus, the vertices of our diagram are $(s, s)$, $(i, s)$, $(s, i)$, and $(i, i)$.
For two edges of this diagram their direction is determined by the fact
that $i$ is a GS-manipulation: namely, both of the edges adjacent to $(s, s)$
are directed away from $(s, s)$. Thus, by renaming the players if necessary, 
any $2$-by-$2$ GS-game for any voting rule can be represented by one of the six diagrams in Figure~\ref{fig:2-by-2-mm}.   
\begin{figure}
\begin{center}
\resizebox{11cm}{!}{\includegraphics{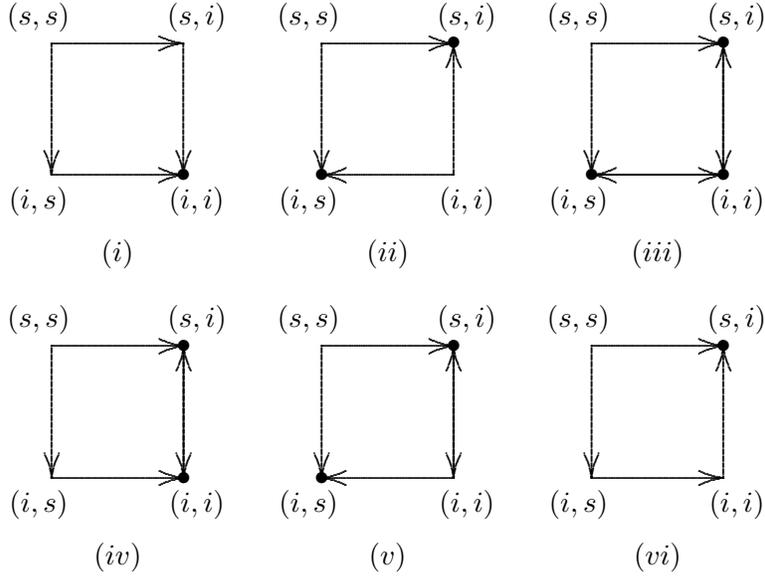}}
\end{center}
\vspace{-5mm}
\caption{Diagrams for $2$-by-$2$ GS-Games}\label{fig:2-by-2-mm}
\end{figure}
Observe that an action profile in a $2$-by-$2$ game is a Nash equilibrium
if and only if the corresponding vertex in the diagram of the game has
two incoming edges. 

The following proposition is immediate.

\begin{proposition}
Every $2$-by-$2$ GS-game has at least one Nash equilibrium.
\end{proposition} 

\begin{proof}
Since both insincere strategies are GS-manipulations we have arrows from $(s,s)$ both to $(i,s)$ and $(s,i)$. For these not to be Nash equilibria both arrows from $(i,s)$ to $(i,i)$ and from $(s,i)$ to $(i,i)$ must be directed towards $(i,i)$. Bit then $(i,i)$ is a Nash equilibrium.
\end{proof}

On the six diagrams in Figure~\ref{fig:2-by-2-mm} Nash equilibria are marked by black dots.

\begin{example}\label{ex:plurality-ii}
{\em
Consider the GS-game for the preference profile $ (v_1,v_2,v_3,v_4)=(abc, bac, cab, cba)$
under the Plurality voting rule, with ties broken according to $a>b>c$.
In this game players $1$ and $2$ are the GS-manipulators; 
their GS-manipulations are $v_1^* = v_1[a;b]$
and $v^*_2=v_2[b;a]$, which result in election of $b$ and $a$, respectively. Note that, if both GS-manipulators
vote insincerely, $c$ remains the election winner. Thus, this game
corresponds to diagram (ii) in Figure~\ref{fig:2-by-2-mm}.}
\end{example}

We will say that a diagram $D$ in Figure~\ref{fig:2-by-2-mm} is {\em realisable} as a diagram of a GS-game by a voting rule $\calR$ on $n$-voter profiles
if there exists a preference profile $V$ consisting of at most $n$ voters and a $2$-by-$2$ game $G\in\GS(V, \calR)$
such that $D$ is the diagram for $G$. We also say that $D$ is {\em realisable} by a voting rule $\calR$ if there exists 
a profile $V$ (without restriction on the number of voters) such that $D$ is the diagram for~$G$. 

Our next goal is to understand which diagrams are realisable by the voting rules we have chosen to concentrate on in this paper, i.e., $k$-Approval voting rules. Me found that most of the diagrams are already relaized by Plurality and the remaining ones by 2-Approval. 

\begin{theorem}
The only diagrams realizable by Plurality are (ii), (iii), (iv) and (v).
\end{theorem}

\begin{proof}
Consider a profile $V$, and assume that voters $1$ and $2$ are the only GS-manipulators in $V$.
Suppose that the Plurality winner in $V$ is $w$, voter $1$ manipulates in favor of $a$, 
and voter $2$ manipulates in favor of $b$. Since $1$ and $2$ are GS-manipulators, 
we have $w\neq\tp(v_1)$, $w\neq\tp(v_2)$, and we can assume that the GS-manipulations
of voters $1$ and $2$ are given by $v_1^*=v_1[\tp(v_1);a]$, $v_2^*=v_2[\tp(v_2);b]$.
Let $V^1=(V_{-1},v^*_1)$, $V^2=(V_{-2},v^*_2)$, $V^{1,2}=(V^1_{-2},v^*_2)$.

The winner in $V^{1,2}$ can be $w$, $a$, or $b$, and the case where the winner is $w$ corresponds to diagram (ii). 
Further, if $a=b$, then $a$ is the winner at $V^1$, $V^2$, and $V^{1,2}$; this corresponds to diagram (iii).
Thus, suppose that the winner at $V^{1,2}$ is $a$ or $b$ and $a\neq b$.
This means that one of the arrows adjacent to $(i, i)$ must be bidirectional, ruling out the three diagrams (i), (ii) and~(vi). 
  We have thus proved that diagrams (i) and (vi) are not realizable by Plurality.
  
Example~\ref{ex:plurality-ii} shows how to realize diagram (ii). We now construct examples for the remaining three cases. 
Diagram (iii) can be realized in profile $V=(cawb, bawc, wabc)$ with ties broken according to $w>a>b>c$. The winner at $V$ is $w$, and both the first and second players can manipulate in favor of $a$, which is therefore the winner at all manipulated profiles $V^1$, $V^2$, and $V^{1,2}$.

Consider now the profile $V=(dabwc, cbawd, wbacd)$ with tie-breaking order $w>a>b>c>d$. Candidate $w$ is the winner in $V$, Voter $v_1$ manipulates in favor of $a$ (horisontal arrow), and $v_2$ in favor of $b$. If both players manipulate, the result is still $a$ by the tie-breaking order. Since $v_1$ prefers $a$ to $b$ this example realises diagram (iv).


Diagram (v) can be obtained on profile $V=(bacw, cbaw, wbac)$, and use $a>w>b>c$ as a tie-breaking rule. Candidate $w$ is the winner at $V$, candidate $a$ is winning at $V^{1}$ and $V^{1,2}$, and candidate $b$ is winning at $V^2$, however this time $v_1$ prefers $b$ to $a$ and hence regrets his choice of manipulating.
\end{proof}

\begin{theorem}
Diagrams (i) and (vi) are both realizable by 2-Approval voting rule.
\end{theorem}

\begin{proof}
In both cases we will consider the alphabet tie-breaking, i.e., the tie-breaking order is $a>b>\ldots >y>z$.
\begin{description}
\item  Diagram (i). Let $V =(v_1,v_2,v_3) = (xywa\ldots, ztwb\ldots,cd\ldots)$. The 2-Approval winner at $V$ is $c$. The first two voters
are the GS-manipulators with manipulations $v_1^*=v_1[xy;aw]$ and $v_2^*=v_2[zt;bw]$ in favor of $a$ and $b$, respectively. Further, at $(v_1^*, v_2^*, v_3)$ the 2-Approval winner is $w$ which both manipulators prefer over $a$ and $b$.

\item Diagram (vi). Let $V =(v_1,v_2,v_3)=(xywa\ldots, ztbaw\ldots,cd\ldots)$. The 2-Approval winner at $V$ is $c$. The first two voters
are the GS-manipulators with manipulations $v_1^*=v_1[xy;aw]$ and $v_2^*=v_2[zt;bw]$ in favor of $a$ and $b$, respectively.
At $(v_1^*, v_2^*, v_3)$ the 2-Approval winner is $w$, in which case voter 2 would regret manipulating as he'd rather prefer $a$ than $w$. 
This realizes diagram~$(vi)$.
\end{description}
\vspace{-0.7cm}
\end{proof}

We note that the second manipulator in the realisation of diagram (vi) have chosen an illogical manipulation promoting $w$ and not promoting $a$. We will call such manipulations {\em unsound}. Soundness of manipulations  will be important issue later. 

We note that for the Borda rule all six diagrams are realizable as well (see~\cite{ElkindEtAlIJCAI2015} for the proof).

\subsection{Manipulator/Countermanipulator Games}

{\color{black}
In this section we consider voting manipulation games where there are two strategic players, one of which is a GS-manipulator and another is a countermanipulator.   A GS-manipulator has a GS-manipulation that change the outcome of the election in his favour; and some other player, the ``countermanipulator''  cannot manipulate but can perform a countermanipulation that `neutralises' to some extent the action of the manipulator (such as voter 4 in Example~\ref{ex1}. We call such games Manipulator/Countermanipulator games. We focus on $2$-by-$2$ games of this type---i.e., one manipulator and one countermanipulator. 
There are six possible forms of normal games that can be achieved in a $2$-by-$2$ Manipulator/Countermanipulator game, these are forms (i)-(vi)  presented in Figure~2. As before, we denote by $s$ the sincere action of the player and by $i$ his insincere action. 
}

\begin{theorem}
Only the forms (i)-(vi) are possible to emerge as diagrams of $2$-by-$2$ Manipulator/Counter Manipulator games. All six diagrams are realisable under 2-Approval. Only (iii) is realisable for two voters under Plurality and (vi) can be realised with more than two voters under Plurality.
\end{theorem}

\begin{figure}
\label{fig:2-by-2-mc}
\begin{center}
\resizebox{12cm}{!}{\includegraphics{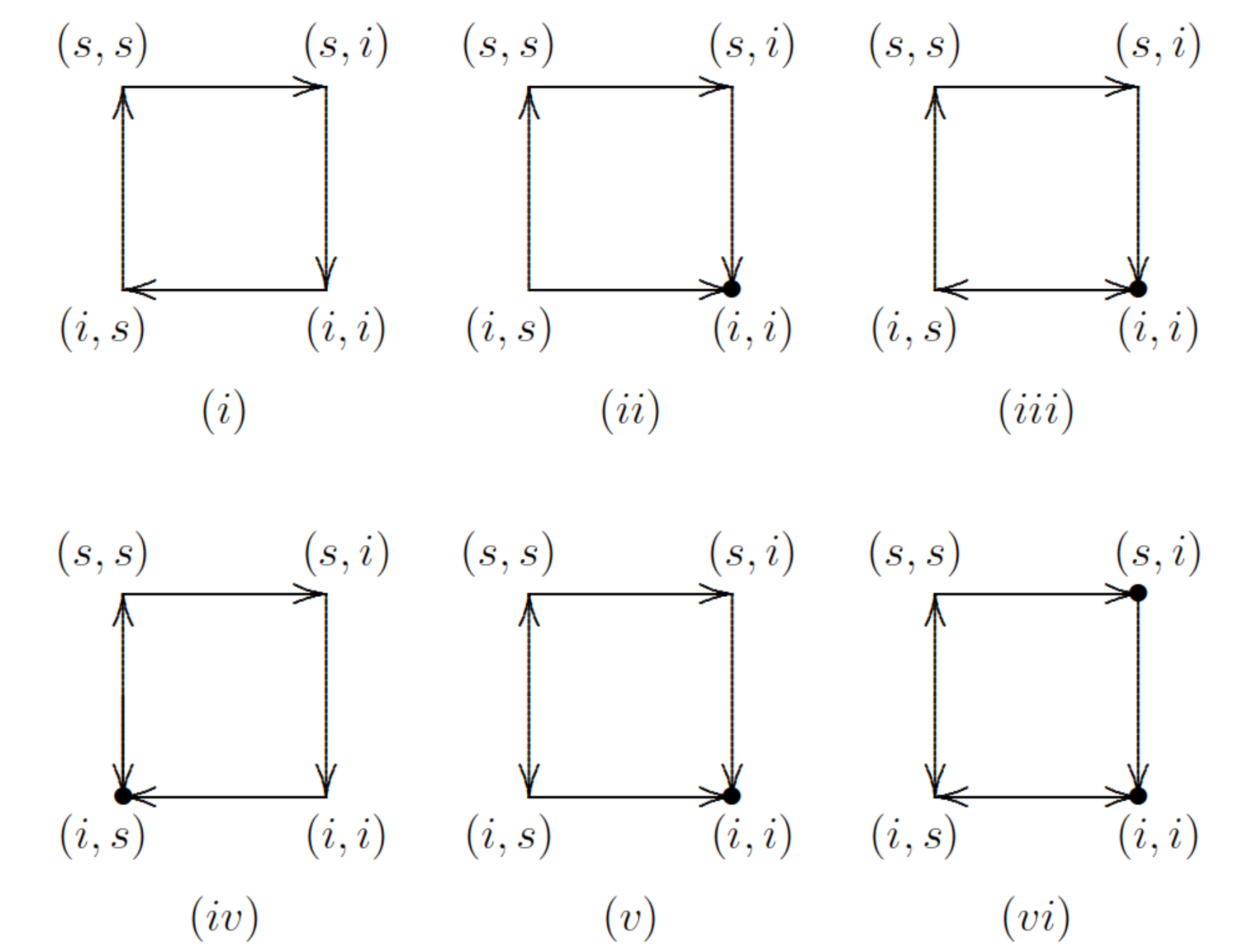}}
\end{center}
\vspace{-5mm}
\caption{Diagrams for $2$-by-$2$ Manipulator/Countermanipulator Games}
\end{figure}

\begin{proof}
Let us introduce the relation $c\sqsupset_V c'$ on candidates $c,c'\in C$ which means that at profile $V$ either $c$ has higher Plurality score than $c'$ or they have equal scores but $c> c'$.

\textbf{Plurality:} It can easily be seen that only one form is possible with two voters under Plurality. Suppose our two voters, 1 and 2, have favourite candidates $a$ and $w$ respectively, with $w$ being higher on the tie-breaking order so $w$ wins (note, if both have the same favourite neither can be a manipulator). Then voter 1 may manipulate in favour of a new candidate $x$ which beats $w$ on tie-break by submitting it first instead of $a$. The only way 2 can countermanipulate is promoting a new candidate $y$ that beats $x$ on tie break. Then $y>x>w>a$. If 1 plays sincerely, but 2 countermanipulates,  then $y$ wins again but 2 would regret countermanipulating. Thus we get diagram (iii)  and it is the only possible diagram realisable by Plurality on 2-voter profiles. 

If there are more than two voters, then diagram ($vi$) is also realisable but no other. Again suppose $a$ wins with $s$ points at sincere profile $V$ and that voter 1 has a manipulation in favour of some candidate $x$, where $x$ originally either also had $s$ points with $a>x$ or $x$ had $s-1$ points with $x>a$, obtaining a profile $V_1$ with $x$ as the winner. Then voter 2, as the counter manipulator, can either promote a new candidate, $y$, to beat $x$, or promote $a$ and make him winner again at profile $V_2$ depending on which is the best candidate in favour of whom voter 2 can countermanipulate.

If $a$ was not voter  2's best candidate in favour of whom he can countermanipulate but~$y$, then $V=(bx\ldots, cy\ldots, \ldots)$, $x$ wins at $V_1=(xb\ldots, cy\ldots, \ldots)$,  and  $y$  wins at $V_2=(xb\ldots, yc\ldots, \ldots)$. Then  $y\sqsupset_{V_2} x\sqsupset_{V_2} a$ and hence $y\sqsupset_V x$. If 1, at $V_2$, reverses his manipulative vote leading to $V_3=(bx\ldots, yc\ldots, \ldots)$, then $y$ would still win at $V_3$ since $y\sqsupset_{V_3} a $ (their positions are the same as in $V_2$) and $a\sqsupset_{V_3}b$ (their positions are the same as in $V$. But voting for $y$ is not a manipulation for voter 2. Hence we have a game with diagram (iii).

If $a$ was voter 2's best candidate in favour of whom he can countermanipulate, then $V=(bx\ldots, ca\ldots, \ldots)$, $x$ wins at $V_1=(xb\ldots, ca\ldots, \ldots)$,  and after voter 2 promoted $a$ she becomes unbeatable at $V_2=(xb\ldots, ac\ldots, \ldots)$ and also at $V_3=(bx\ldots, ac\ldots, \ldots)$. We therefore have diagram (vi). 

\textbf{2-Approval:} Below in Table~\ref{table:examples} are examples of profiles for each diagram under 2-Approval voting with two voters (apart from (iii) which has already been relised with two voters under Plurality). For all profiles the tie-breaking rule is $x>e>a>b>c>d>u$ and $a$ wins under sincere votes. Voter 1 is the manipulator (column player) and voter 2 is the countermanipulator (row player). This proves the theorem.\qedhere

\begin{table}[!h]
  \centering
  	\quad
    \begin{tabular}{rclc}
    \toprule
    Diagram & Voter & Preferences &        Strategic votes \\
    \midrule
    \multirow{2}{*}{($i$)} & 1     & $c\, d \,|\, e\,x\,b\,a$ & $v_1[e;d]$ \\
     & 2     & $a\,b \,|\, d\, x\, c\,e$ & $v_2[ab;dx]$ \\
    \midrule
    \multirow{2}{*}{($ii$)} & 1     & $c\, d \,|\, e\,u\,x\,b\,a$ & $v_1[cd;eu]$ \\
     & 2     & $a\,b \,|\,u\, x\, d\, c\,e$ & $v_2[ab;ux]$ \\
    \midrule
    \multirow{2}{*}{($iv$)} & 1     & $b\, a \,|\, u\,c\,d$ & $v_1[a;u]$ \\
     & 2     & $c\,d \,|\,u\,a\,b$ & $v_2[d;u]$ \\
    \midrule
    \multirow{2}{*}{($v$)} & 1     & $c\, d \,|\, u\,e\,x\,b\,a$ & $v_1[cd;eu]$ \\
     & 2     & $a\,b \,|\,u\, x\, d\, c\,e$ & $v_2[b;u]$ \\
    \midrule
    \multirow{2}{*}{($vi$)} & 1     & $b\, a \,|\, u\,c\,d$ & $v_1[a;u]$ \\
     & 2     & $c\,d \,|\,a\,u\,b$ & $v_2[d;a]$ \\
    \bottomrule
    \end{tabular}
    \caption{Realisations of diagrams (i),(ii),(iv)-(vi) under 2-Approval voting rule.}
  \label{table:examples}
\end{table}
\end{proof}

We note that apart from the first diagram all the remaining ones have at least one Nash equilibrium.

\section{Nash Equilibria}\label{sec:nash}

One of the most important characteristics of any game is whether it has Nash equilibria (NE) in pure strategies or not. Various schools of thought provide different opinions on whether or not players will end up choosing equilibrium strategies. We are not plunging into these debates, we simply consider the existence of an NE as an important feature of the game. 

We saw that in the presence of countermanipulators we cannot expect a voting manipulation game to have a NE. Indeed, the manipulator is happy, when the countermanipulator is unhappy and vice versa. However, in the case of a GS-game, when all strategic players are manipulators, things are not so obvious. 

In this section, we study the existence of Nash equilibria in GS-games for $k$-Approval with $k=1,2,3,4$. We start with the easiest case which is Plurality.

\subsection{Voting Manipulation Games under Plurality}

{\color{black} 
We will first show that for Plurality, a Nash equilibrium always exists for every voting manipulation game, even in presence of countermanipulators.

\begin{theorem}\label{thm:ne-plurality}
Any voting manipulation game under Plurality has a Nash equilibrium in pure strategies.
\end{theorem}

\begin{proof}
Fix a profile $V$ amd let $w$ be the Plurality winner at $V$ with score $t$. Let $S\subset C$ be the set of candidates such that either their score is $t$ and they lose $w$ on the tie-break or their score is $t-1$ and they are higher than $w$ on the tie-breaking order. In particular, $w\in S$. Let $S^+\subset S$ be the set of candidates in favour of whom there is a GS-manipulation.  Let $p\in S^+$ be the candidate such that $p\sqsupset_V q$ for all $q\in S^+$ and suppose  voter i who can manipulate in favour of $p$ with $v^*_i=v_i[\text{top}(v_i);p]$. Suppose, first, that there is another voter $k$ who can manipulate in favour of $p$ with $v^*_k=v_k[\text{top}(v_k);p]$. Then it is easy to check that $((V_{-i}, v^*_i)_{-k},v^*_k)$ is a Nash equilibrium with winner $p$. If this second manipulator in favour of $p$ does not exist, no GS-manipulator (existing at $V$) can change the result at $(V_{-i}, v^*_i)$ by his actions due to the choice of $p$. However  a countermanipulator can possibly change the result countermanipulating in favour of $q\in S\setminus S^+$ (exactly such situation has occurred in Example~\ref{ex1}). Note that his top preference is not $p$ thus he cannot make the score of $p$ lower. Let $q\in S\setminus S^+$ be maximal with respect to $\sqsupset_V$ for which some voter, say voter $j$ with $v^*_j$, can countermanipulate. Then $((V_{-i}, v^*_i)_{-j},v^*_j)$ is a NE.
\end{proof}
}
\subsection{Manipulation strategies for $k$-Approval. First observations.}

To get more insight into GS-games for $k\ge 2$ we have to understand possible manipulation strategies of players.
%
Note, first, that under $k$-Approval any GS-manipulation of voter $i$ is equivalent to a vote of the form $v_i[X; Y]$, where $X\subseteq\tp_k(v_i)$, $Y\subseteq C\setminus\tp_k(v_i)$.
There are two types of GS-manipulators for $k$-Approval voting: 
those who rank the current winner~$w$ in top $k$
positions, and those who do not. We give a more precise description in the following lemma.

\begin{lemma}
\label{mantypes}
Let the voting rule be $k$-Approval for $k\ge 1$. Let $V$ be a profile and $w$ be the winner at $V$. Let also $x$ be an alternative, other than $w$.
Then any manipulation in favour of $x$ at $V$ falls under one and only one of the following two categories: 
\begin{description}
\item[Type 1] A voter $i$ increases the score of $x$ by 1 without decreasing the score of  $w$. In this case both $w$ and $x$ are not approved by the manipulator with $x\succ_i w$, the manipulator moves $x$ to the first $k$ positions leaving $w$ not approved. We refer to $i$ as a {\em promoter} of~$x$. 
\item[Type 2] A voter $i$ reduces the score of $w$ and possibly the scores of some other alternatives by 1 without increasing the score of $x$. In this case both $w$ and $x$ are approved by the manipulator with $x\succ_i w$, the manipulator removes $w$ from his top $k$ positions leaving $x$ there. We refer to $i$ as a {\em demoter} of~$w$. 
\end{description}
Moreover, for $k=1$ only manipulations of type~1 are possible.
\end{lemma}

\begin{proof}
If voter $i$ can increase the score of an alternative $x$ in favour of which he manipulates, then it was not approved by this voter originally. He, however, must rank it higher than $w$ (otherwise this is not a GS-manipulation). Thus, $w$ was not approved either and voter $i$ cannot decrease its score. 
The second possibility is that the voter cannot increase the score of the alternative $x$ in favour of which he manipulates since it is already approved by him. This means that he is left with reducing the scores of some of its competitors including the current winner~$w$. This means $w$ is also in the approval set with $x\succ_i w$. This type of manipulation cannot happen for $k=1$ since both $w$ and $x$ must be in the approval set.
\end{proof}

\subsection{2-Approval GS-games with unrestricted manipulations}

We will start with a seemingly discouraging example which shows that for 2-Approval voting GS-games no NE may exist (but watch how unnatural the strategies are).

\begin{example} 
\label{ex:no_eq-2-Approval} 
\normalfont Consider three players, $v_1$, $v_2$, $v_3$, with the following preferences under the 2-Approval voting rule. The tie-breaking rule is $n> a >m > x > b > c> d > e > f$. 
Let
\begin{equation*}
V=(v_1,v_2,v_3)=(ab|cdefnmx,\ cd|nmxbaef,\ ef|bmxnacd).
\end{equation*}
%
%
Under this profile $a$ wins. Voter $v_1$ is not a manipulator as his most preferred candidate, $a$, wins. Both $v_2$ and $v_3$ can manipulate in favour of $n$ and $b$, respectively. 
We define the strategy sets $A_2=\{s_2, i_{2}, i_{2}'\}$ and $A_3=\{s_3, i_{3}, i_{3}'\}$, where $s_2$, $s_3$ are the sincere strategies of $v_2$ and $v_3$ respectively. Here $i_2$ and $i_2'$ are  manipulations of voter 2 where he swaps $c$ and $d$ with $n$ and $m$, and $n$ and $x$, respectively, i.e.,  
\[
i_2=v_2[cd;nm],\qquad \text{and}\qquad i'_2=v_2[cd;nx].
\]
Similarly, $i_3$ and $i_3'$ are manipulations of voter 3 where he swaps $e$ and $f$ with $b$ and $m$ and $b$ and $x$, respectively, submitting 
\[
i_3=v_3[ef;bm],\qquad \text{and}\qquad i'_3=v_3[ef;bx].
\]
To see that this game has no Nash equilibrium we have to consider all 9 strategy profiles that can be realised. These are summarised in Table~\ref{StrategyProfiles} where we give the winner in the given profile, which voter has an incentive to change their strategy and the change in strategy that would be favourable for that voter.
\begin{table}[H]
  \centering
    \begin{tabular}{cccc}
    \toprule
    Profile & Winner  & Voter & Change in Strategy \\
    \midrule
    ($s_2$, $s_3$) & $a$   & $3$ & $s_3$ to $i_3$ \\
    ($s_2$, $i_3$) & $b$   & $2$ & $s_2$ to $i_2$ \\
    ($s_2$, $i_3'$) & $b$   & $2$ & $s_2$ to $i_2'$ \\
    ($i_2$, $s_3$) & $n$   & $3$ & $s_3$ to $i_3'$ \\
    ($i_2$, $i_3$) & $m$   & $3$ & $i_3$ to $i_3'$ \\
    ($i_2$, $i_3'$) & $b$   & $2$ & $i_2$ to $i_2'$ \\
    ($i_2'$, $s_3$) & $n$   & $3$ & $s_3$ to $i_3$ \\
    ($i_2'$, $i_3$) & $b$   & $2$ & $i_2'$ to $i_2$ \\
    ($i_2'$, $i_3'$) & $x$   & $3$ & $i_3'$ to $i_3$ \\
    \bottomrule
    \end{tabular}%
   \caption{Profitable Deviations at Each of the 9 Strategy Profiles}
  \label{StrategyProfiles}%
\end{table}%
A graphical representation of this is in Figure~\ref{fig:daniel_graph}. As we can see, instead of a NE, we have an attractor in the form of a 4-cycle.
\begin{figure}
\begin{center}
\resizebox{8cm}{!}{\includegraphics{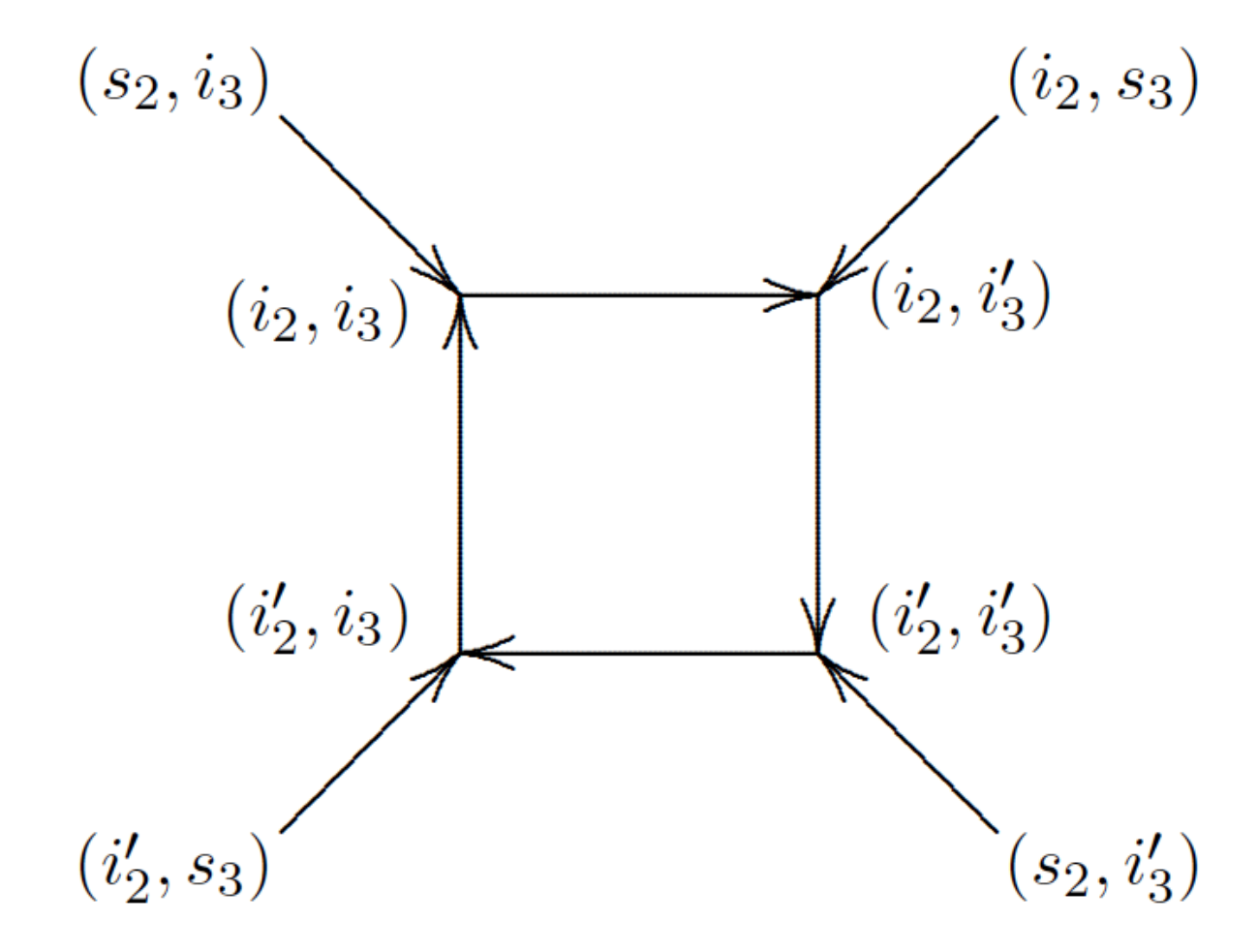}}
\end{center}
\vspace{-5mm}
\caption{Diagram for the Game in Example~\ref{ex:no_eq-2-Approval}}\label{fig:daniel_graph}
\end{figure}

As can be seen, in every profile it is favourable for at least one of the manipulators to change their strategy, hence no Nash equilibrium exists.
\end{example}

This is a no Nash equilibrium game with the minimal number of voters (three), since for two voters, as we will see, every game has a Nash equilibrium. To show this we need the following lemma. 

\begin{lemma}
\label{same_tops}
In any GS-game under the 2-Approval voting rule all Type 2 manipulators must have the same two top preferences and in the same order. In particular, they manipulate in favour of the same candidate.
\end{lemma}

\begin{proof}
First note that all Type 2 manipulators approve the winner $w$ of the sincere profile $V$ but place her in the second position. Also, if a voter is a Type 2 manipulator, its other approved candidate must either be on the same score as $w$ and be lower in tie-breaking order or be on the score one less than $w$ and be higher in the tie-breaking order.
Suppose there are two Type 2 manipulators with different top preferences, call them $v_1$ and $v_2$ with $a_1=\text{top}(v_1)\ne \text{top}(v_2)=a_2$, i.e., $v_1=a_1w\ldots$ and $v_2=a_2w\ldots$. 

When $v_1$ manipulates, demoting $w$ and promoting some other candidate, $a_1$ wins. This means, in particular, $a_1\sqsupset_V a_2$. But similarly, we get $a_2\sqsupset_V a_1$.
This is a contradiction, showing that all Type 2 manipulators must have the same top two preferences (and in the same order).
\end{proof}

\begin{theorem}
Every 2-Approval GS-game with two voters has a Nash equilibrium in pure strategies.
\end{theorem}
\begin{proof}
If only one voter is a manipulator then there is a trivial Nash equilibrium (the manipulator plays insincerely). Hence assume both players, $v_1$ and $v_2$, are manipulators. Due to Lemma~\ref{mantypes}, both voters cannot be  manipulators of Type 1 since the winner of the sincere profile $V$ must be in $\text{top}_2(v_1)\cup \text{top}_2(v_2)$. If both are Type 2 manipulators, by Lemma~\ref{same_tops} their manipulations are in favour of the same candidate so when either plays insincerely a Nash equilibrium is achieved.

So assume $v_1$ is a Type 2 manipulator and  $v_2$ is a Type 1. That is $v_1$'s second best candidate wins at the sincere profile and their top candidate comes second. But $v_2$ is a Type 1 manipulator so if he manipulates in favour of candidate $x$, then $x$ must beat after that both of $v_1$'s top two candidates (or be $v_1$'s) so $v_1$'s manipulation will not change the outcome, giving a Nash equilibrium when just $v_2$ votes insincerely. 
\end{proof}

One of the notable characteristics of the Example~\ref{ex:no_eq-2-Approval}  is that the manipulators have more than one manipulation in their strategy sets (in this case they have two each). We would hypothesise that if manipulators are restricted to no more than one manipulation each (not an unreasonable assumption as to consider more than one manipulation would not be expected of most voters) then a Nash equilibrium should exist in any such voting manipulation game. This has already been shown to be true for $2$-by-$2$ GS-games. 

\subsection{2-Approval GS-games under sound manipulations}

{\color{black}

We saw in Example~\ref{ex:no_eq-2-Approval}  that if voters choose an `irrational' sets of strategies, then NE may not exist. Indeed, both voters can realise that $i_2''=v^*_2[d;n]$ and $i_3''=v^*_3[f;b]$ are much better strategies to use. For both players,  promoting not only the candidate in favour of whom they are manipulating but also $m$  is dangerous as she can inadvertently win (and she does).

Let us now state some additional assumptions that will be used in the following sections.

\begin{definition}
Under $k$-Approval a Type 1 GS-manipulation $v_i[Y,X]$ of voter $i$  in favour of $x\in X$, where $Y\subseteq \text{top}_k(v_i)$ and $X\subseteq C\setminus \text{top}_k(v_i)$    is called \emph{sound} if all the alternatives  which are moved up together with $x$ are preferred to $x$, i.e., for all $x'\in X$, different from $x$, we have that $x'>_i x$.
\end{definition}
}

The reason for this definition is as follows. A voter is manipulating in favour of $x$ which is the highest alternative which he can get manipulating alone. However it may be possible if two or more voters manipulating in unison may secure a better outcome. Hence a voter may wish, just in case, to increase the score of an alternative different from $x$, but increasing scores of alternatives which are below $x$ may end up in promoting them as the new winners. We therefore formulate the following assumption:

\begin{assumption}[Soundness assumption] Let $V$ be a profile and $G$ be a GS-game for $V$ under $k$-Approval voting rule. If voter $i$ has a manipulation strategy in $A_i$ in favour of an alternative $x$, then $A_i$ also contains a sound manipulation strategy in favour of $x$.
\end{assumption}


{\color{black}

\begin{theorem}\label{thm:ne-2app}
Any GS-game for 2-Approval voting rule, 
whose strategy sets satisfy the Soundness Assumption, has a Nash equilibrium. 
\end{theorem}

\begin{proof}
Let $V$ be a profile and $w$ be a winner for 2-Approval voting at $V$. Then $w$ has a maximal score, say $s$. Let $Q\subset C$ be the set of candidates in favour of whom a GS-manipulation at $V$ exists. Then all candidates from $Q$ have scores $s$ or $s-1$. Suppose manipulations of Type 1 exist at $V$ and $q$ be the highest candidate from $Q$ relative to $\sqsupset_V$ in favour of whom manipulation of Type 1 exists. Due to the Soundness assumption, every voter who can manipulate in favour of $q$ has a sound manipulation strategy in favour of her. Suppose one of them, say voter $i$, has the manipulation  $v_i^*=v_i[y;q]$, in which he moves up only $q$.  Then we claim that $V^1=(V_{-i},v_i^*)$
is a Nash equilibrium.  At $V^1$ we have $q \sqsupset_{V^1} w  \sqsupset_{V^1} \text{top}(v_\ell)$ for any voter $\ell$.
Indeed, no manipulation of Type 2 will change the result at $V^1$. Such a manipulation for voter $\ell$ has the form $v^*_\ell=v_\ell[w;a]$, where $a$ loses to $\text{top}(v_\ell)$ in $(V_{-\ell},v^*_\ell)$ (in favour of whom he manipulates), and, hence in $V^1$.  And $\text{top}(v_\ell)$ loses to $q$ in $V^1$.
Any manipulation of Type 1 in favour of $p\in Q$ will not change the result either since $q \sqsupset_V p$, hence $p$ cannot overtake $q$ in $V^1$.

Suppose now that any manipulator of Type 1 in favour of $q$ in his sound manipulation also promotes another candidate whom he ranks higher than $q$, say his sound manipulation is  $v_i^*=v_i[xy;qr]$ with $r\succ_i q$. Then, as we showed before, no other manipulation can lead to winning of $p\in Q$ different from $q$. However, it may lead to winning of $r$ in case voter $j$ manipulates with  $v_j^*=v_j[zt;pr]$ with $r\succ_j p$. In such a case $r\sqsupset_{V^2} q$, where  $V^2=(V^1_{-j},v_j^*)$. Then $V^2$ may be a Nash equilibrium since both voters $i$ and $j$ will be satisfied. The only case, when it is not a Nash equilibrium is when another voter, say $k$, can also manipulate in favour of $q$, say with $v_k^*=v_k[uv;qs]$, where $s\ne r$. Then someone else, say voter $m$, may be in position to vote  $v_m^*=v_m[ef;pr]$ making $r$ winner again. After the first two manipulations, only $q$ and $r$ can stay in contention (all other candidates will have not enough points for this. All who manipulated would not want to revert to their sincere votes. Eventually a Nash equilibrium will be reached.

It remains to consider the case when only Type 2 manipulators exist. In this case by Lemma~\ref{same_tops} they all manipulate in favour of the same candidate and, if one of them manipulates, all others are happy.
\end{proof}

As Example~\ref{ex:no_eq-2-Approval}  showed, the Soundness Assumption is necessary for the existence of Nash equilibria. 
}  

\subsection{$3$-Approval}

A stronger rationality assumption that we will consider is that of restricting the set of manipulation strategies to those that take a minimal number of changes.


\begin{definition}
A manipulation of Type 1 in favour of $x$ under $k$-Approval is {\em minimal} if $x$ is the only alternative which is moved up and to the $k$th position while the alternative which formerly occupied $k$th position becomes not approved.
\end{definition}
 In particular, a Type 1 minimal manipulation it is always sound. 
 
 \begin{definition}
 A manipulation of Type 2 is called {\em minimal} if the smallest number, say $\ell$, of alternatives are moved down while the best $\ell$ not previously approved alternatives are moved up. 
\end{definition}
Observe that in a  GS-game there may be several minimal manipulation available to the same player. We can now formulate the following assumption. 

\begin{assumption}[Minimality assumption (MA)] 
Let $V$ be a profile and $G$ be a GS-game for $V$. Then $G$ satisfies the Minimality Assumption if for every $j=1,2,\ldots, n$ all GS-manipulations in the strategy set $A_j$ are minimal.
\end{assumption}

To secure a Nash equilibrium for 2-Approval GS-games we had to assume Soundness assumption. A similar result holds for $3$-Approval, though under the stronger Minimality Assumption. 

We start with the following lemma.

\begin{lemma}
\label{at_most__two}
In any GS-game under the 3-Approval voting rule the set of candidates in favour of whom Type 2 manipulators can manipulate has cardinality at most~2.
\end{lemma}

\begin{proof}
Suppose that the set of Type 2 manipulators at profile $V$ is not empty and let $Q$ be the set of candidates in favour of whom they can manipulate. Let $w$ be the winner at $V$ and let $q\in Q$ be such that $q\sqsupset_V q'$ for all $q'\in Q$. Let $p\in Q$ be the candidate such that $p\sqsupset_V q'$ for all $q'\in Q$ such that $q'\ne q$. 

We claim that $Q\subseteq \{q,p\}$. Indeed, for any third alternative $r\in C$ to win as a result of Type 2 manipulation, the respective voter must demote $w,q,p$, which is impossible  since $r$ by Lemma~\ref{mantypes} must be among the approved candidates. 
\end{proof}

\begin{theorem}\label{thm:ne-3app}
Under the Minimality Assumption any GS-game for 3-Approval has a Nash equilibrium in pure strategies.
\end{theorem}

\begin{proof}
We can assume that the set of GS-manipulators $N(V, 3\text{-App})$ is non-empty. Let $w$ be a 3-Approval winner at the sincere profile $V$ with the score $\scr_3(w,V)=t$. By Lemma~\ref{mantypes}, we can partition the set of GS-manipulators into a set of \emph{promoters}, i.e., the set of $j \in N(V, 3\text{-App})$ such that $w\not\in \tp_3(v_j)$, and a set of \emph{demoters}, i.e., such voters $v_\ell$ for whom $w \in \tp_3(v_\ell)$.

Assume first that the set of promoters is non-empty. Let $Q$ be the set of candidates in favour of whom promoters can manipulate, and let $\bar{p}\in Q$ such that $\bar{p}\sqsupset_V q$ for all $q\ne \bar{p}$.
Let therefore $i$ be a promoter in favour of $\bar{p}$, and let $v_i^*$ be his GS-manipulation strategy. 
We now show that $V^*=(V_{-i},v_i^*)$ is a NE. 
Observe that by minimality of $v_i^*$ and by definition of $\bar{p}$ no other promoter can change the outcome of $V^*$. We can therefore focus on the set of demoters.  Let $j$ be a demoter and let $v_j^*$ be its manipulation strategy in favour of candidate $p$. 
By minimality assumption, $v_j^*$ either removes only $w$ from $\tp_3(v_j)$, or removes $w$ together with a second candidate {\color{black} (and the third candidate in $\tp_3(v_j)$ by Lemma~\ref{mantypes} must be then $p$)}. While the first case would not be a profitable deviation at $V^*$ since the result of the election does not change, we need more attention in the second case {\color{black} since $\bar{p}$ could be that second demoted candidate by voter $j$.

The fact that voter $j$ had to demote $\bar{p}$, due to the minimality assumption, means that $\bar{p}\sqsupset_V p$. In this case we will also have  $\bar{p}\sqsupset_{V^{**}} p$, where $V^{**}=(V^*_{-j},v^*_j)$.} 
%
Hence $\bar p$ wins against $p$ in $V^{**}$, and $v_j^*$ is not a profitable deviation for $j$ at $V^*$.

We can now assume that the set of promoters is empty, and that therefore all GS-manipulators in $N(V, 3\text{-App})$ are demoters. 
{\color{black}
By Lemma~\ref{at_most__two} there are at most two candidates in favour of whom manipulation is possible. Let us denote them $p_1$ and $p_2$ with $p_1\sqsupset_V p_2$. 
}

The set of demoters can be partitioned into a set $V_1$ of GS-manipulators for $p_1$, whose minimal strategy is to lower the current winner $w$ only, and a set $V_2$ of GS-manipulators for $p_2$, whose minimal strategy is to lower both $w$ and $p_1$. Note that voters in $V_2$ has $p_2$ as their top candidate.

If $V_2=\emptyset$, all manipulators manipulate in favour of $p_1$ and Nash equilibrium obviously exist. 
Consider then the case in which $V_1$ and $V_2$ are both non-empty.  Note that once a voter from $V_2$ manipulates, $p_1$ can no longer win no matter how other voters vote. 

For all pairs of candidates $x,y$ different than $w, p_1,$ or $p_2$, let 
$$
V_2^{x,y}=\{j\in V_2\mid v_j[w,p_1; x,y]\in A_j \}.
$$
A voter $v\in V_1$ with manipulation $v^*=v[w,x]$ ranks $p_2$ lower than $x$ and this manipulation may serve as a countermanipulation strategy to manipulations of voters in $V_2^{x,y}$, making, under certain circumstances $x$ to win instead of $p_2$, so we need to design a strategy profile in which such situations do not occur. 
Let therefore 
\[
V_1^x=\{j\in V_1\mid v_j[w;x]\in A_j \text{ and } 3\text{-app}(V_{-\{i,j\}},v_i^*, v_j^*)=x \text{ for some  $i\in V_2^{x,y}$}\},
\]
 i.e., $V_1^x$ is the set of voters who have a countermanipulation move in favour of $x$ when $x$'s score is being raised by a manipulator in favour of $p_2$ by some voter in $V_2$.
 
If there exists $j\in V_2$ and $x,y$ such that $j\in V_2^{x,y}$ but both $V_1^x$ and $V_1^y$ are empty, then it is easy to see that $(V_{-j}, v_j[w,p_1;x,y])$ is a Nash equilibrium: voters in $V_1$ cannot change the outcome, and voters in $V_2$ are satisfied with having $p_2$ the winner.
Suppose then that this is not the case, i.e., for each pair of candidates $x,y$, either $V_2^{x,y}$ is empty, or one of $V_1^{x}$ and $V_1^y$ are not empty. 
Pick one voter from each non-empty $V_1^x$ -- they are all distinct since each voter belongs to at most one $V_1^x$, having a single manipulation strategy. Without loss of generality let them be $J=\{1,\dots,k\}$, and let $V^*=(V_{-J}, v_1^*,\dots, v_k^*)$ be the profile in which all GS-manipulators in $J$ play their manipulation strategies $v_1^*,\dots, v_k^*$. 
Since $p_1$ is the winner in $V^*$, all voters in $V_1$ do not have incentives to deviate.
Voters in $V_2$ also do not have incentives to deviate. For if any $j\in V_2^{x,y}$ manipulate in $V^*$ the result would change in favour of either $x$ or $y$, which by construction are less preferred by $j$ than $p_1$.
This concludes the proof.
\end{proof}

\subsection{$4$-Approval}

In contrast, for $4$-Approval the existence of Nash equilibria is no longer guaranteed, even if GS-manipulations are restricted to minimal ones.

\begin{theorem}
There exists a game $G = (V, 4\text{-App}, (A_i)_{i\in N(V, 4\text{-App})})$,
where for each player $i\in N(V, 4\text{-App})$ the set $A_i$ consists of $i$'s truthful vote and $i$'s minimal GS-manipulation, such that $G$ has no Nash equilibrium.
\end{theorem}

\begin{proof}
Let $\{u_1, u_2, u_3, v_1, v_2, v_3\}$ be a set of voters using 4-Approval to choose one among 8 candidates $\{w,d_1,d_2,d_3,d,e,c,x\}$. 
Let the tie-breaking rule be $w\succ d_1 \succ d_2 \succ d_3 \succ c \succ d \succ e \succ x$. 
Let $V$ be the profile in Table~\ref{No_Nash_4-App}, where the top four approved candidates are those above the line in each individual preference:
%
\begin{table}[H]
  \centering
    \begin{tabular}{c|c|c||c|c|c}
    $u_1$& $u_2$ & $u_3$   & $v_1$ & $v_2$ & $v_3$ \\
    \hline
    $w$&$c$&$d$&$d_1$&$d_2$&$d_3$\\
    $d_1$&$d_2$&$d_3$&$w$&$w$&$w$\\
    $e$&$d$&$c$&$d$&$d_1$&$d_2$\\
    $c$&$d_3$&$d_2$&$d_3$&$x$&$d_1$\\
    \hline
    $d$&$w$&$w$&$c$&$c$&$d$\\
    $x$&$e$&$d_1$&$d_2$&$d$&$x$\\
    $d_2$&$x$&$e$&$e$&$d_3$&$e$\\
    $d_3$&$d_1$&$x$&$x$&$e$&$c$\\
    \bottomrule
    \end{tabular}%
   \caption{A Profile for the 4-Approval GS-game with no Nash Equilibrium}
  \label{No_Nash_4-App}%
\end{table}
The scores of alternatives are as follows: all of $w, d_1, d_2, d_3$ get 4 points, $c$ and $d$ get 3 points, $e$ and $x$ gets 1 point. The winner at $V$ is therefore $w$.
The first three candidates cannot manipulate: $u_1$ ranks the winner $w$ on top; $u_2$ and $u_3$ rank $w$ just below the approval line and hence there is no candidate they prefer to $w$ that can be promoted.
Thus the set of GS-manipulators $N(V,4\text{-App})=\{v_1,v_2,v_3\}$. We restrict the set of strategies of voter $j$ to the sincere strategy $s_j$ and the minimal manipulation strategy $i_j$, hence $A_j=\{s_j,i_j\}$, where
\begin{itemize}
\item $i_1=v_1[w;c]$ making $d_1$ winner;
\item $i_2=v_2[wd_1;cd]$ making $d_2$ winner;
\item $i_3=v_3[wd_2d_1;dxe]$ making $d_3$ winner.
\end{itemize}

There are 8 strategy profiles in this game, and in Table~\ref{table:nonash} we indicate for each strategy profile which of the candidates is the winner and which of the voters have an incentive to change the strategy. 

\begin{table} 
\begin{center}
\begin{tabular}{cccc}
Profile & Winner & Deviation & New Winner\\
\hline
$(s_1,s_2,s_3)$ & $w$  & $1$ switches to $i_1$ & $d_1$\\
$(i_1,s_2,s_3)$ & $d_1$  & $3$ switches to $i_3$ & $d_3$\\
$(s_1,i_2,s_3)$ & $d_2$  & $1$ switches to $i_1$ & $c$\\
$(s_1,s_2,i_3)$ & $d_3$  & $2$ switches to $i_2$ & $d$\\
$(i_1,s_2,i_3)$ & $d_3$  & $2$ switches to $i_2$ & $c$\\
$(s_1,i_2,i_3)$ & $d$  & $3$ switches to $s_3$ & $d_2$\\
$(i_1,i_2,s_3)$ & $c$  & $2$ switches to $s_2$ & $d_1$\\
$(i_1,i_2,i_3)$ & $c$  & $1$ switches to $s_1$ & $d$
\end{tabular}
\caption{Deviations from strategy profiles.}
\label{table:nonash}
\end{center}
\end{table}

At every strategy profile there is at least one player that prefers the winner of a different profile to the current one. Hence, there is no Nash equilibrium. The response dynamics is illustrated on the following diagram, drawn in line with the representation of 2-by-2 games we presented in Section~\ref{sec:2by2}. We omit indices in Figure~\ref{fig:cube} as it does not create any confusion. The cycle
$(s,i,s)\rightarrow (i,i,s)\rightarrow (i,s,s)\rightarrow (i,s,i)  \rightarrow (i,i,i)\rightarrow (s,i,i)\rightarrow (s,i,s)$
marked in the figure with a dashed line shows that there is no NE in the game.
\begin{figure}
\begin{center}
\resizebox{8cm}{!}{\includegraphics{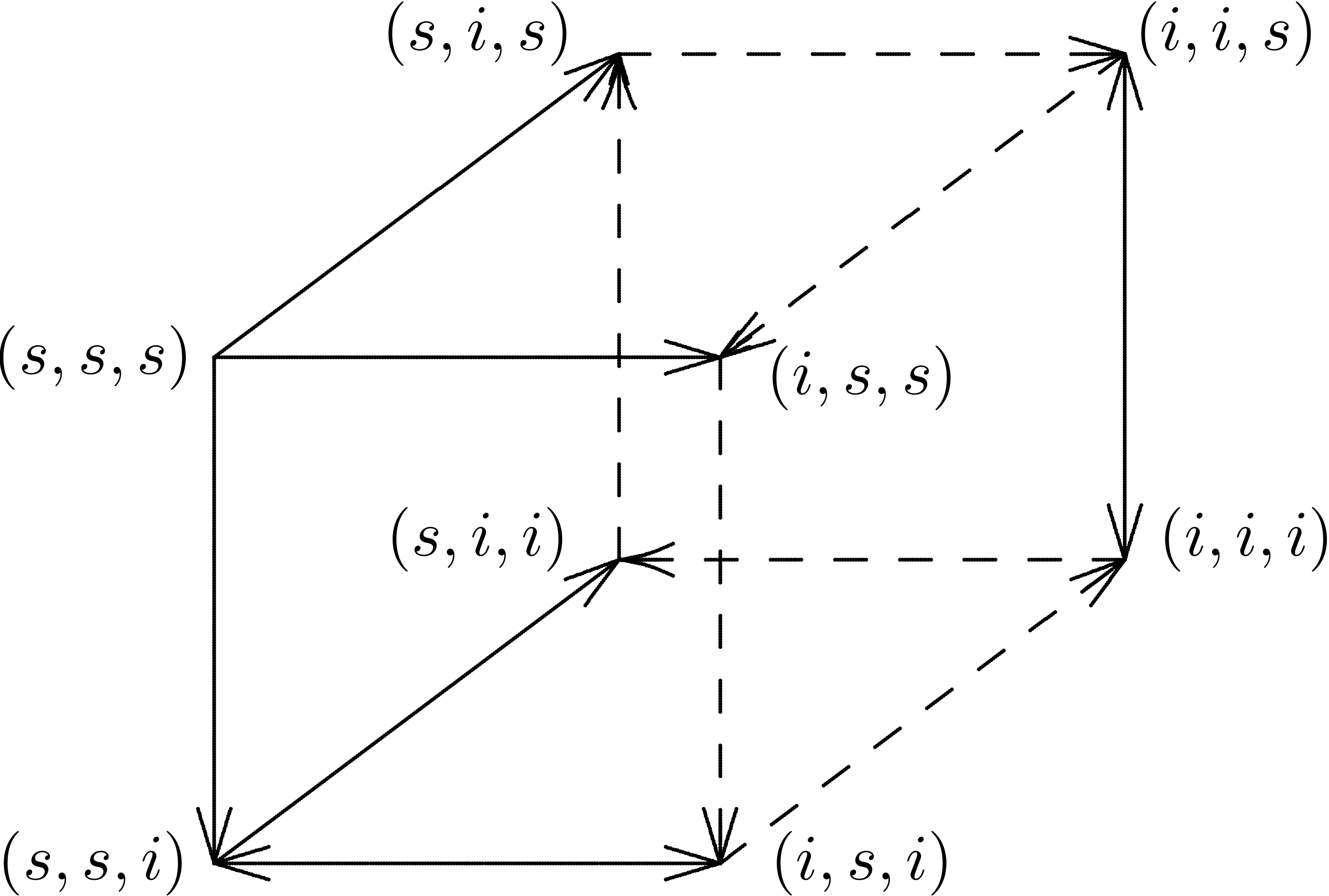}}
\end{center}
\caption{A 4-Approval GS-game with no NE.}\label{fig:cube}
\end{figure}
\end{proof}

\section{Discussion of the Results}\label{sec:discussion}

In this paper we suggested a new framework for studying voting manipulation games, and in this section we clarify some aspects of the proposed framework and provide additional justification for our hypothesis. 

{\bf Boundedly rational voters.} Unlike \cite{mye-web:j:voting} and their followers, in our approach voters are boundedly rational and cannot see beyond manipulations in Gibbard-Satterthwaite sense, and countermanipulations to those. To contemplate a countermanipulation voters must have a higher degree of rationality (at least level-2 in the cognitive hierarchy model by \citet{camerer2004cognitive}) so most of the time we assume that the game is played by Gibbard-Satterthwaite manipulators alone. In particular, they vote sincerely if they cannot change the result. This hypothesis is more in line with what real voters do.  It is extremely hard to estimate how many 
strategic voters are present in a given election, but the percentage of those who actually manipulated is easier to 
estimate. For instance, \cite{KW2013} estimate the number of such voters, called misaligned,
 in Japanese elections between 2,5\% and 5,5\%. 
Moreover, \cite{BBS2013} show that preference misrepresentation is related to cognitive skills, and \cite{CKMS2014} 
demonstrate that decision-making ability in laboratory experiments correlates strongly with socio-economic status 
and wealth. Therefore, it is reasonable to assume that only a small fraction of voters in an election would act 
strategically when given an opportunity to do so.

{\bf Complete information.} The Gibbard-Satterthwaite theorem is proved under the assumption that voters know sincere preferences of other voters.  It is thus natural to make the same assumption when we investigate the interaction of Gibbard-Satterthwaite manipulators. 

{\bf Randomised strategies.} Firstly, we would like to emphasise our intention to stay within the ordinal model of preferences. Thus we do not assume that voters can identify the utilities that they will enjoy when each candidate wins.  If we assumed that voters were allowed to include randomised strategies in their strategy sets, then outcomes become lotteries and  voters must be able to compare them. However, within ordinal model of preferences our voter with preferences $a>b>c$ is unable, for example, to compare the lottery that gives 50\% chance of  $a$ and 50\% chance of  $c$ with the sure thing lottery that gives him 100\% chance of $b$. 

{\bf Tie-breaking.} Our results clearly depend on the way ties are broken, and the alphabetic tie-breaking is the most popular method for doing this. It violates neutrality but preserves anonymity which usually considered as a more important property. Randomised tie-breaking is not applicable in our framework. It would be interesting, however, to see if our main results survive under other methods of deterministic tie-breaking.

{\bf Other Models of Bounded Rationality.} The voters in our games are boundedly rational. In particular, they believe that a voter who cannot manipulate or countermanipulate must stay sincere; they do not think they can have any other more sophisticated intentions. 
This is similar to assumptions of the paper \cite{EGRS} where they model a view of the game that a level-2 voter from the cognitive hierarchy model might have. Our voters are even less sophisticated than level-2 voters since the latter can think of manipulating or countermanipulating.


\section{Conclusions}\label{sec:conclusions}

To model the games played by voters in elections more realistically we must take in consideration those voters are only boundedly rationality. 
We suggest here such a model. Our voters assume that other players may be strategic and can either manipulate or countermanipulate and we study the games that they have to play as a result. 
Unfortunately, there are just a few known characteristics of games that can be meaningfully investigated. The most important of them is whether a particular game has Nash equilibria in pure strategies or not. We, thus, analyse voting manipulation games from this perspective. 

We noted that the existence of countermanipulators prevents any chances of having Nash equilibria (so these games have to be investigated from another perspective). Hence we concentrate on games played by Gibbard-Satterthwaite manipulators only (GS-games).

We have initiated the study of such games. We have shown that 
for Plurality these games exhibit a fairly simple structure; however, already for  $k$-Approval
with $k>1$ GS-games are quite complicated, and it may therefore be difficult
for players to decide on their actions. 
%
Many questions concerning GS-games remain open. The most immediate of them is to fully understand 
the role of minimality assumptions in the proof for 3-Approval. Further afield, it would be interesting
to extend our study to other voting rules, most notably Borda, and to identify reasonable restrictions 
on the manipulators' strategy spaces that lead to existence  of Nash equilibria or 
 make it easy to compute manipulations that weakly dominate truthtelling.

\bibliographystyle{abbrvnat}
\bibliography{voting}
\end{document}